\input ./style/arxiv-general.cfg
\documentclass[aos,MSNbibl,seceqn,dvips]{arximspdf}
\makeatletter
   \@ifpackageloaded{graphicx}{}{\usepackage{graphicx}}
\makeatother
\usepackage{url,breakurl}

\doi{10.1214/15-AOS1337}
\volume{43}
\issue{5}
\pubyear{2015}
\firstpage{2055}
\lastpage{2085}
\docsubty{FLA}

\makeatletter

\def\eqd{\,{\stackrel{d}{=}}\,}
\def\geqd{\,{\stackrel{d}{\ge}}\,}

\newtheorem{theorem}{Theorem}

\newtheorem{lemma}{Lemma}

\newproclaim{assumption}{Assumption}
\newproclaim{definition}{Definition}
\newproclaim{remark}{Remark}
\newproclaim{example}{Example}
\newcommand{\eqref}[1]{(\ref{#1})}

\newcommand{\diag}{\operatorname{diag}}
\newcommand{\sign}{\operatorname{sign}}

\def\R{\mathbb{R}}
\newcommand{\ident}[0]{\mathbf{I}}
\newcommand{\iidsim}{\,\stackrel{\mathsf{i.i.d.}}{\sim}\,}

\newcommand{\X}{\mathbf{X}}
\newcommand{\y}{\mathbf{y}}
\renewcommand{\b}{\bolds{\beta}}
\newcommand{\e}{\mathbf{z}}
\newcommand{\sig}{\sigma}
\newcommand{\Xproxy}{\tilde{\X}}
\newcommand{\thr}{T}
\newcommand{\swap}{\mathsf{swap}}
\newcommand{\Wset}{\mathcal{W}}
\makeatother

\begin{document}
\begin{frontmatter}

\title{Controlling the false discovery rate via knockoffs}
\runtitle{FDR control via Knockoffs}

\begin{aug}
\author[A]{\fnms{Rina Foygel}~\snm{Barber}\thanksref{T1}\ead[label=e1]{rina@uchicago.edu}}
\and
\author[B]{\fnms{Emmanuel J.}~\snm{Cand\`es}\thanksref{T2}\ead
[label=e2]{candes@stanford.edu}}
\runauthor{R.~F. Barber and E.~J. Cand\`es}
\affiliation{University of Chicago and Stanford University}
\thankstext{T1}{Supported in part by NSF Grant DMS-12-03762.}
\thankstext{T2}{Supported in part by AFOSR  Grant
FA9550-09-1-0643,  NSF  Grant CCF-0963835 and by the Math $+$ X
Award from the Simons Foundation.}
\address[A]{Department of Statistics\\
University of Chicago\\
Chicago, Illinois 60637\\
USA\\
\printead{e1}}

\address[B]{Department of Statisitcs\\
Stanford University\\
390 Serra Mall\\
Stanford, California 94305\\
USA\\
\printead{e2}}
\end{aug}

%
\received{\smonth{5} \syear{2014}}
%
\revised{\smonth{4} \syear{2015}}

%
\begin{abstract}
In many fields of science, we observe a response variable together
with a large number of potential explanatory variables, and would
like to be able to discover which variables are truly associated
with the response. At the same time, we need to know that the false
discovery rate (FDR)---the expected fraction of false discoveries
among all discoveries---is not too high, in order to assure the
scientist that most of the discoveries are indeed true and
replicable. This paper introduces the \textit{knockoff filter}, a new
variable selection procedure controlling the FDR in the statistical
linear model whenever there are at least as many observations as
variables. This method achieves exact FDR control in finite sample
settings no matter the design or covariates, the number of variables
in the model, or the amplitudes of the unknown regression
coefficients, and does not require any knowledge of the noise
level. As the name suggests, the
method operates by manufacturing knockoff variables that are
cheap---their construction does not require any new data---and are
designed to mimic the correlation structure found within the
existing variables, in a way that allows for accurate FDR control,
beyond what is possible with permutation-based methods. The method
of knockoffs is very general and flexible, and can work with a broad
class of test statistics. We test the method in combination with
statistics from the Lasso for sparse regression, and obtain
empirical results showing that the resulting method has far more
power than existing selection rules when the proportion of null
variables is high.
\end{abstract}

%
\end{frontmatter}

\section{Introduction}

Understanding the finite sample inferential properties of procedures
that select and fit a regression model to data is possibly one of the
most important topics of current research in theoretical
statistics. This paper is about this problem and focuses on the
accuracy of variable selection in the classical linear model under
arbitrary designs.

\subsection{The false discovery rate in variable selection}
\label{sec:fdr}

Suppose we have\break recorded a response variable of
interest $y$ and many potentially explanatory variables $X_j$ on $n$
observational units. Our observations obey the classical linear
regression model
%
\begin{equation}
\label{eqn:model} \y=\X\b+\e,
\end{equation}
where as usual, $\y\in\R^n$ is a vector of responses,
$\X\in\R^{n\times p}$ is a known design matrix, $\b\in\R^p$ is an
unknown vector of coefficients and $\e\sim\mathcal{N}(0,\sig
^2\ident)$ is
Gaussian noise. Because we are interested in valid inference from
finitely many samples, we shall mostly restrict our attention to the
case where $n \ge p$ as otherwise the model would not even be
identifiable. Now in modern settings, it is often the case that there
are typically just a few relevant variables among the many that have
been recorded. In genetics, for instance, we typically expect that
only a few genes are associated with a phenotype $y$ of interest. In
terms of the linear model \eqref{eqn:model}, this means that only a few
components of the parameter $\b$ are expected to be nonzero. While
there certainly is no shortage of data fitting strategies, it is not
always clear whether any of these offers real guarantees on the
accuracy of the selection with a finite sample size. In this paper, we
propose  controlling the false discovery rate (FDR) among all the
selected variables, that is, all the variables included in the model, and
develop novel and very concrete procedures, which provably achieve
this goal.

Informally, the FDR is the expected proportion of falsely selected
variables, a false discovery being a selected variable not appearing
in the true model. Formally,\vspace*{1pt} the FDR of a selection procedure
returning a subset $\hat S \subset\{1, \ldots, p\}$ of variables is
defined as
%
\begin{equation}
\label{eqn:fdr} \mathsf{FDR} = \mathbb{E} \biggl[ \frac{\#\{j :
\beta_j = 0 \mbox{
and } j \in\hat S\}}{
\#\{j : j \in\hat S\}\vee1} \biggr].
\end{equation}
(The definition of the denominator above sets the fraction
to zero in the case that zero features are selected,
i.e., $\hat{S}=\varnothing$; here we use the notation $a\vee
b=\max
\{a,b\}$.)
We will say that a selection rule controls
the FDR at level $q$ if its FDR is guaranteed to be at most $q$ no
matter the value of the coefficients $\b$. This
definition asks to control the type I error averaged over the selected
variables and is both meaningful and operational. Imagine we have a
procedure that has just made 100 discoveries. Then roughly speaking,
if our procedure is known to control the FDR at the 10\% level, this
means that we can expect at most 10 of these discoveries to be false
and, therefore, at least 90 to be true. In other words,
if the collected data were the outcome of a
scientific experiment, then we would expect that most of the variables
selected by the knockoff procedure correspond to real effects
that could be reproduced in follow-up experiments.

In the language of hypothesis testing, we are interested in the $p$
hypotheses $H_j : \beta_j = 0$ and wish to find a multiple comparison
procedure able to reject individual hypotheses while controlling the
FDR. This is the reason why we will at times use terminology from this
literature, and we may say that $H_j$ has been rejected to mean that feature
$j$ has been selected, or we may say that the data provide evidence
against $H_j$
to mean that the $j$th variable likely belongs to the model.

\subsection{The knockoff filter}
\label{sec:new}

This paper introduces a general FDR controlling procedure that is
guaranteed to work under \textit{any} fixed design $\X\in\R
^{n\times p}$,
as long as $n> p$, and the response $\y$ follows a linear Gaussian
model as in \eqref{eqn:model}. An important feature of this procedure is
that it does not require any knowledge of the noise level
$\sigma$. Also, it does not assume any knowledge about the number of
variables in the model, which can be arbitrary. We now outline the
steps of this new method.

\textit{Step} 1: \textit{Construct knockoffs}. For each feature $\X_j$ in
the model (i.e., the $j$th column of $\X$), we construct a ``knockoff''
feature $\Xproxy_j$. The goal of the knockoff variables is to imitate
the correlation structure of the original features in a very specific
way that
will allow for FDR control.

Specifically, to construct the knockoffs, we first calculate the Gram matrix
$\bolds{\Sigma}=\X^\top\X$ of the original features,\setcounter{footnote}{2}\footnote{We assume
throughout that $\bolds\Sigma$ is invertible as the model would
otherwise not be identifiable.} after normalizing each feature such
that ${\Sigma}_{jj}= \|\X_j\|^2_2 = 1$ for all $j$. We will ensure that
these knockoff features obey
%
\begin{equation}
\label{eqn:construct2} \Xproxy^\top\Xproxy= \bolds{\Sigma},\qquad \X ^\top
\Xproxy= \bolds{\Sigma} - \diag\{\mathbf{s}\},
\end{equation}
where $\mathbf s$ is a $p$-dimensional nonnegative vector. In words,
$\Xproxy$ exhibits the same covariance structure as the original
design $\X$, but in addition, the correlations between distinct
original and knockoff variables are the same as those between the
originals (because $\bolds\Sigma$ and $\bolds\Sigma-\diag\{\mathbf
s\}$ are
equal on off-diagonal entries),
\[
\X_j^\top\Xproxy_k = \X_j^\top
\X_k \qquad\mbox{for all }j\neq k.
\]
However, comparing a feature $\X_j$ to its
knockoff $\Xproxy_j$, we see that
\[
\X_j^\top\Xproxy_j = \Sigma_{jj}-
s_j = 1- s_j,
\]
while $\X_j^\top\X_j =\Xproxy_j^\top\Xproxy_j =1$.
To ensure that our
method has good statistical power to detect signals, we will see that
we should choose the entries of $\mathbf s$ as large
as possible so that a variable $\X_j$
is not too similar to its knockoff $\Xproxy_j$.

A strategy for constructing $\Xproxy$ is to choose $\mathbf s \in
\mathbb{R}^p_+$ satisfying $\diag\{\mathbf{s}\}\preceq2\bolds
\Sigma$, and
construct the $n \times p$ matrix $\Xproxy$ of knockoff features as
%
\begin{equation}
\label{eqn:construct} \Xproxy= \X \bigl(\ident- \bolds\Sigma ^{-1}\diag\{
\mathbf{s}\} \bigr) + \tilde{\mathbf U} {\mathbf C};
\end{equation}
here, $\tilde{{\mathbf U}}$ is an $n \times p$ orthonormal matrix that is
orthogonal\footnote{In this version of the construction, we are
implicitly assuming $n\geq2p$. Section~\ref{sec:p_or_2p} explains how to
extend this method to the regime $p< n <2p$. } to the span of
the features $\X$, and ${\mathbf C}^\top{\mathbf C} = 2\diag\{
\mathbf{s}\}-
\diag\{\mathbf{s}\} \bolds\Sigma^{-1} \diag\{\mathbf{s}\}$ is a Cholesky
decomposition (whose existence is guaranteed by the condition
$\diag\{\mathbf s\}\preceq2\bolds\Sigma$; see Section~\ref{sec:2p}
for details).

\textit{Step} 2: \textit{Calculate statistics for each pair of original and
knockoff variables}. We now wish to introduce the statistics $W_j$ for
each $\beta_j \in\{1, \ldots, p\}$, which will help us tease apart
those variables that are in the model from those that are
not. These $W_j$'s are constructed so that large positive values are
evidence against the null hypothesis $\beta_j = 0$.

In this instance, we consider the Lasso model
\cite{tibshirani1996regression}, an $\ell_1$-norm penalized regression
that promotes sparse estimates of the coefficients $\b$, given by
%
\begin{equation}
\label{eqn:lasso} \hat{\b}(\lambda) = \mathop{\operatorname{arg min}}_{\mathbf b} \biggl\{
\frac{1}{2}\| \y- \X\mathbf{b}\|^2_2+\lambda\|
\mathbf{b}\|_1 \biggr\}.
\end{equation}
For sparse linear models, the Lasso is known to be asymptotically
accurate for both variable selection and for coefficient or signal
estimation (see, e.g.,
\cite{bickel2009simultaneous,yanivLASSO,zhang2008sparsity,zhao2006model}),
and so even in a nonasymptotic setting, we will typically see
$\hat{\b}(\lambda)$ including many signal variables and few null
variables at some value of the penalty parameter $\lambda$.
Consider
the point $\lambda$ on the Lasso path at which feature
$\X_j$ first enters the model,
%
\begin{equation}
\label{eqn:lassoZ} \mbox{Test statistic for feature $j$}= \sup \bigl\{\lambda:\hat{
\beta}_j(\lambda) \neq0 \bigr\},
\end{equation}
which is likely to be large for most of the signals, and small for
most of the null variables. However, to be able to quantify this and
choose an appropriate threshold for variable selection, we need to use
the knockoff variables to calibrate our threshold. With this in mind,
we instead compute the statistics in \eqref{eqn:lassoZ} on the
augmented $n
\times2p$ design matrix $[\X\ \Xproxy]$ (this is the columnwise
concatenation of $\X$ and $\Xproxy$), so that $[\X\ \Xproxy]$ replaces
$\X$ in \eqref{eqn:lasso}. This yields a $2p$-dimensional vector $(Z_1,
\ldots, Z_p$, $ \tilde{Z}_1, \ldots, \tilde{Z}_p)$. Finally, for each
$j \in\{1, \ldots, p\}$, we set
%
\begin{equation}
\label{eqn:W} W_j = Z_j \vee\tilde{Z}_j
\cdot\cases{ +1, &\quad $Z_j > \tilde{Z}_j$,\vspace*{2pt}
\cr
-1, &\quad $Z_j < \tilde{Z}_j$}
\end{equation}
(we can set $W_j$ to zero in case of equality $Z_j = \tilde{Z}_j$). A
large positive value of $W_j$ indicates that variable $\X_j$ enters
the Lasso model early (at some large value of $\lambda$) and that it
does so before its knockoff copy $\Xproxy_j$. Hence this is an
indication that this variable is a genuine signal and belongs in the
model. We may also consider other alternatives for constructing the
$W_j$'s: for instance, instead of recording the variables' entry into
the Lasso model, we can consider forward selection methods and record
the order in which the variables are added to the model; see
Section~\ref{sec:Wstats} for this and other alternatives.

In Section~\ref{sec:mainresults}, we discuss a broader methodology,
where the
statistics $W_j$ may be defined in any manner that satisfies the
sufficiency property and the antisymmetry property, which we will
define later on; the construction above is a specific instance that we
find to perform well empirically.

\textit{Step} 3: \textit{Calculate a data-dependent threshold for the
statistics}. We wish to select variables such that $W_j$ is large
and positive, that is, such that $W_j \ge t$ for some $t > 0$. Letting $q$
be the target FDR, define a data-dependent threshold $\thr$ as
%
\begin{equation}
\label{eqn:thresh} \thr=\operatorname{min} \biggl\{t\in\Wset:\frac{\#\{
j:W_j \leq-t\}}{\#\{
j:W_j \geq t\}\vee1}\leq q \biggr
\}
\end{equation}
or $\thr=+\infty$ if this set is empty, where
$\Wset=\{|W_j|:j=1,\ldots,p\}\setminus\{0\}$ is the set of unique
nonzero\footnote{If $W_j=0$ for some feature $\X_j$, then this gives
no evidence for rejecting the hypothesis $\beta_j=0$, and so our
method will never select such variables. } values attained by the
$|W_j|$'s. We shall see that the fraction appearing above is an
estimate of the proportion of false discoveries if we are to select
all features $j$'s with $W_j \geq t$. For this reason, we will often
refer to this fraction as the \textit{knockoff estimate of FDP}.

For a visual representation of this step, see
Figure~\ref{fig:square}, where we plot the point $(Z_j,\tilde Z_j)$
for each
feature $j$, with black dots denoting null features and red squares
denoting true signals. Recall that $W_j$ is positive if the original
variable is selected before its knockoff (i.e., $Z_j>\tilde Z_j$) and
is negative otherwise (i.e., $Z_j<\tilde Z_j$). Therefore a feature $j$
whose point lies below the dashed diagonal line in Figure~\ref{fig:square}
then has a positive value of $W_j$, while points above the diagonal
are assigned negative $W_j$'s. For a given value of $t$, the numerator
and denominator of the fraction appearing in \eqref{eqn:thresh}
above are given by the numbers of
points in the two gray shaded regions of the figure (with nulls and
nonnulls both counted, since in practice we do not know which
features are null).

\begin{figure}

\includegraphics{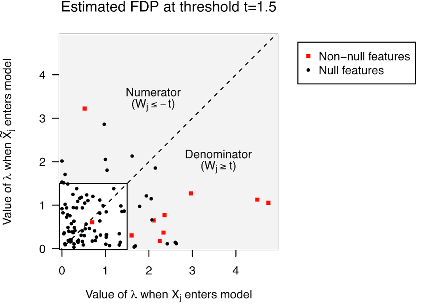}

\caption{Representation of the knockoff procedure plotting pairs
$(Z_j, \tilde Z_j)$. Black dots correspond to the null hypotheses
($\beta_j = 0$) while red squares are nonnulls ($\beta_j \neq0$).
Setting $t=1.5$, the number of points in the
shaded region below the diagonal is equal to $\#\{j : W_j \geq t\}$,
the number of selected variables at this threshold, while the number of
points in the shaded region above the diagonal is equal to $\#\{j :
W_j \leq- t\}$. Observe that the true signals (red squares) are primarily
below the diagonal, indicating $W_j>0$, while the null features (black
dots) are roughly symmetrically distributed across the diagonal.}
\label{fig:square}
\end{figure}

With these steps in place, we are ready to define our procedure:

\begin{definition}[(Knockoff)]\label{def:knockoff}
Construct $\Xproxy$ as in \eqref{eqn:construct}, and calculate
statistics $W_j$ satisfying the sufficiency and antisymmetry
properties [defined in Section~\ref{sec:mainresults}; \eqref{eqn:W}~above gives
an example of a statistic satisfying these properties]. Then select
the model
\[
\hat{S} = \{j: W_j\geq\thr\},
\]
where $\thr$ is the data-dependent threshold \eqref{eqn:thresh}.
(Note that $\hat{S}$ and $T$ both implicitly depend
on the choice of the target FDR level $q$.)
\end{definition}

A main result of this paper is that this procedure controls a
quantity nearly equal to the FDR:

\begin{theorem}\label{thm:knockoff}
For any $q\in[0,1]$, the knockoff method satisfies
\[
\mathbb{E}\biggl[{\frac{\#\{j:\beta_j=0\mbox{ and }j\in\hat{S}\}}{\#\{
j:j\in\hat{S}\} + q^{-1}}}\biggr] \leq q,
\]
where the expectation is taken over the Gaussian noise $\e$ in the
model \eqref{eqn:model},
while treating $\X$ and $\Xproxy$ as fixed.
\end{theorem}

The ``modified FDR'' bounded by this theorem is very close to the FDR
in settings where
a large number of features are selected (as adding $q^{-1}$ in the denominator
then has little effect), but it sometimes may be preferable to control
the FDR exactly. For this, we propose a slightly more conservative procedure:

\begin{definition}[(Knockoff$+$)]\label{def:knockoffplus}
Select a model as in Definition~\ref{def:knockoff} but with a
data-dependent threshold $\thr$ defined as
%
\begin{equation}
\label{eqn:threshplus} \thr=\operatorname{min} \biggl\{t\in\Wset :\frac{1+\#\{j:W_j \leq-t\}
}{\#\{j:W_j \geq t\}\vee1}\leq
q \biggr\}.
\end{equation}
\end{definition}

Note that the threshold $T$ chosen by knockoff$+$ is always
higher (or equal to) than that chosen in \eqref{eqn:thresh} by the
knockoff filter,
meaning that knockoff$+$ is (slightly) more conservative.

Our second main result shows that knockoff$+$ controls the FDR.

\begin{theorem}\label{thm:knockoffplus}
For any $q\in[0,1]$, the knockoff$+$ method satisfies
\[
\mathsf{FDR}=\mathbb{E}\biggl[{\frac{\#\{j:\beta_j=0\mbox{ and }j\in
\hat{S}\} }{\#\{j:j\in\hat{S}\}\vee1}}\biggr] \leq q,
\]
where the expectation is taken over the Gaussian noise $\e$ in model
\eqref{eqn:model},
while treating $\X$ and $\Xproxy$ as fixed.
\end{theorem}

We have explained why a large positive value of ${W}_j$ bears some
evidence against the null hypothesis $\beta_j = 0$, and now give a
brief intuition for how our specific choice of threshold allows
control of FDR (or of the modified FDR). The way in which
$\mathbf W$ is constructed implies that the signs of the $W_j$'s are
i.i.d. random for the ``null hypotheses,'' that is, for those $j$'s
such that $\beta_j=0$. Therefore, for any threshold $t$,
%
\begin{equation}
\label{eqn:eqd}\#\{j:\beta_j=0\mbox{ and }W_j \geq t\}
\eqd\#\{j:\beta_j=0\mbox{ and }W_j \leq-t\},
\end{equation}
where $\eqd$ means equality in distribution. In Figure~\ref
{fig:square}, for
instance, $\#\{j:\beta_j=0\mbox{ and }W_j \geq t\}$ is the number of null
points (black dots) in the shaded region below the diagonal, while
$\#\{j:\beta_j=0\mbox{ and }W_j \leq-t\}$ is the number of null points
in the shaded region above the diagonal. Note that the null points
are distributed approximately symmetrically across the diagonal, as
described by \eqref{eqn:eqd}.

Hence we can estimate
the false discovery proportion (FDP) at the threshold $t$ as
%
\begin{eqnarray}\label{eqn:FDP_intuition}
\frac{\#\{j:\beta_j=0\mbox{ and }W_j \geq t\}}{\#\{j:W_j \geq t\}
\vee
1} & \approx&\frac{\#\{j:\beta_j=0\mbox{ and }W_j \leq- t\}}{\#\{
j:W_j \geq t\}
\vee1}
\nonumber
\\[-8pt]
\\[-8pt]
\nonumber
 & \leq&\frac{\#\{j:W_j\leq- t\}}{\#\{
j:W_j\geq t\}\vee1}=: \widehat{\mathsf{FDP}}(t),
\end{eqnarray}
where $\widehat{\mathsf{FDR}}(t)$ is the knockoff estimate of FDP.
The knockoff procedure can be interpreted as finding a threshold via
$\thr= \operatorname{min} \{t\in\Wset: \widehat{\mathsf{FDP}}(t) \le q\}$,
with the convention that $\thr= + \infty$ if no such $t$ exists; this
is the most liberal threshold with the property that the estimated FDP
is under control. In fact, the inequality in \eqref{eqn:FDP_intuition}
will usually be tight because most strong signals will be selected
before their knockoff copies (in Figure~\ref{fig:square}, we see that
most of
the red squares lie below the diagonal, i.e., $W_j\geq0$); this means
that our estimate of FDP will probably be fairly tight unless the
signal strength is weak. (We will see later that the additional
``$+1$'' appearing in the knockoff$+$ method, yielding a slightly more
conservative procedure, is necessary both theoretically and
empirically to control FDR in scenarios where extremely few
discoveries are made.)

\subsection{Outline of the paper}
\label{sec:outline}

The rest of this paper is organized as follows:
\begin{itemize}
\item In Section~\ref{sec:mainresults}, we introduce the more general
form of
our variable selection procedure
and give some theoretical properties of the procedure that will allow
for FDR control.
\item In Section~\ref{sec:compare} we discuss some related methods and
strategies for FDR control for variable selection.
We compare our proposal to permutation-based methods, to the
Benjamini--Hochberg (BHq) procedure and some of its variants and to
other methods. In particular,
in Section~\ref{sec:simulations}, we present simulations to
demonstrate that
the method is effective in practice and performs well compared to the
BHq and related procedures.
\item In Section~\ref{sec:exper}, we present an application of the knockoff
method to real data
where the task is to find mutations in the HIV-1 protease or reverse
transcriptase that are associated
with drug resistance.
\item In Section~\ref{sec:sequential}, we move to a more general
problem of sequential
hypothesis testing, and we
show that our approach is an example of a procedure for
controlling FDR in a sequential hypothesis testing problem. Some proofs
are deferred to
the supplementary materials \cite{knockoffsupplement}.
\item In Section~\ref{sec:discussion} we close the paper with a discussion
outlining possible
extensions of this work.
\end{itemize}


\section{Knockoffs and FDR control}\label{sec:mainresults}

\subsection{The knockoff features}
\label{sec:proxy}

We begin with the construction of the knockoff features $\Xproxy_j$
and set $\bolds{\Sigma}=\X^\top\X$ as before. We first present the method
in the natural setting where $n \ge2p$ before presenting ways of
extending the construction to the range $p<n < 2p$. To ease
readability, vectors and matrices are boldfaced throughout the paper
whereas scalars are not.

\subsubsection{The natural setting \texorpdfstring{$n\ge2p$}{n>=2p}}\label{sec:2p}

As introduced earlier, the
matrix $\Xproxy$ obeys
%
\begin{equation}
\label{eqn:proxy2} \left[\matrix{\X&\Xproxy}\right]^\top
\left[\matrix{\X&\Xproxy}\right]= \left[\matrix{ \bolds{
\Sigma}&\bolds{\Sigma} - \diag\{\mathbf s\}\vspace*{2pt}
\cr
\bolds{\Sigma} -\diag
\{ \mathbf s\} &\bolds{\Sigma}} \right] := \mathbf G,
\end{equation}
where $\mathbf s\in\R^p$ is some vector. A necessary and sufficient
condition for $\Xproxy$ to exist is that $\mathbf G$ is positive
semidefinite. Indeed, by standard Schur complement calculations,
$\mathbf G\succeq\mathbf0$ if and only if $\diag\{\mathbf s\}\succeq
\mathbf0$
and $2\bolds\Sigma\succeq\diag\{\mathbf s\}$,
as claimed earlier. Now let $\tilde{\mathbf{U}}\in\R^{n\times p}$
be an
orthonormal matrix whose column space is orthogonal to that of $\X$ so
that $\tilde{\mathbf{U}}^\top\X=\mathbf{0}$: such a matrix exists
because $n
\ge2p$. A simple
calculation then shows that setting
%
\begin{equation}
\label{eqn:knockoffEC} \Xproxy=\X \bigl(\ident- \bolds{\Sigma }^{-1}\diag\{
\mathbf s\} \bigr) + \tilde{\mathbf{U}} \mathbf C
\end{equation}
gives the correlation structure specified in \eqref{eqn:proxy2}, where
$\mathbf C^\top\mathbf C = 2\diag\{\mathbf s\}-\diag\{\mathbf s\}
\bolds\Sigma
^{-1}\diag\{\mathbf s\}\succeq\mathbf0$.

Now that we understand the condition on $\mathbf s$ necessary for knockoff
features with the desired correlation structure to exist, it remains
to discuss which one we should construct, that is, to specify a choice
of $\mathbf s$. Returning to the example of the statistic from Section~\ref{sec:new},
we will have a useful methodology only if those variables that truly
belong to the model tend to be selected before their knockoffs as we
would otherwise have no power. Imagine that variable $\X_j$ is in the
true model. Then we wish to have $\X_j$ enter before $\Xproxy_j$. To
make this happen, we need the correlation between $\Xproxy_j$ and the
true signal to be small, so that $\Xproxy_j$ does not enter the Lasso
model early. In other words, we would like $\X_j$ and $\Xproxy_j$ to
be as orthogonal to each other as possible. In a setting where the
features are normalized, that is, $\Sigma_{jj}=1$ for all $j$, we would
like to have $\Xproxy_j^\top\X_j = 1-s_j$ as close to zero as
possible. Below, we consider two particular types of knockoffs:
\begin{itemize}
\item\textit{Equi-correlated knockoffs}: Here, $s_j =
2\lambda_{\operatorname{min}}(\bolds{\Sigma}) \wedge1$ for all $j$, so that
all the correlations take on the identical value
%
\begin{equation}
\label{eqn:set_s_constant} \langle\X_j, \Xproxy_j\rangle= 1- 2
\lambda_{\operatorname{min}}(\bolds{\Sigma}) \wedge1.
\end{equation}
Among all knockoffs with this equi-variant property, this
choice minimizes the value of $|\langle\X_j, \tilde{\X}_j\rangle|$.

\item\textit{SDP knockoffs}: Another possibility is to select knockoffs
so that the average correlation between an original variable and its
knockoff is minimum. This is done by solving the convex problem
%
\begin{equation}
\label{eqn:set_s_sum} \mbox{minimize}\quad \sum_j
(1-s_j) \quad\mbox{ subject to }\qquad 0\leq s_j \leq1, \diag\{
\mathbf s\} \preceq2\bolds\Sigma.
\end{equation}
This optimization problem is a highly structured, semidefinite program
(SDP), which can be solved very efficiently \cite{BoydBook}.
\end{itemize}
If the original design matrix $\X$ allows us to construct a knockoff
matrix $\Xproxy$ such that
the $s_j$'s are high (near $1$), the knockoff procedure will have high
power; if this is not
the case, then the power is likely to be lower.

\subsubsection{Extensions to $p<n<2p$}
\label{sec:p_or_2p}
When $n < 2p$, we can no longer find a subspace of dimension $p$ which
is orthogonal to $\X$, and so we cannot construct $\tilde{\mathbf
{U}}$ as
above. We can still use the knockoff filter, however, as long as the
noise level $\sigma$ is known or can be estimated. For instance,
under the Gaussian noise model \eqref{eqn:model}, we can use the fact
that the residual sum of squares from the full model is distributed as
$\|\y-\X\hat{\b}{}^{\mathsf{LS}}\|^2_2 \sim\sigma^2 \cdot
\chi^2_{n-p}$, where $\hat{\b}{}^{\mathsf{LS}}$ is the vector of\vspace*{1.5pt}
coefficients in a least-squares regression. Now letting
$\hat{\sigma}$ be our estimate of $\sigma$, draw a
$(2p-n)$-dimensional vector $\y'$ with
i.i.d. $\mathcal{N}(0,\hat{\sigma}^2)$ entries. If $n-p$ is large,
then $\hat{\sigma}$ will be an extremely accurate estimate of
$\sigma$, and we can proceed as though $\sigma$ and $\hat{\sigma}$
were equal. We then augment the response vector $\y$ with the new
$(2p-n)$-length vector~$\y'$, and augment the design matrix $\X$ with
$2p-n$ rows of zeros. Then approximately,
\[
\left[\matrix{ \y
\cr
\y'} \right] \sim\mathcal{N}\left(\left[
\matrix{\X
\cr
\mathbf0}\right] \b, \sigma^2\ident\right) .
\]
We now have a linear model with $p$ variables and
$2p$ observations, and so we can apply the knockoff filter to this row-augmented
data using the method described for the $n\geq2p$ setting.
We emphasize that since the knockoff matrix $\Xproxy$ is constructed
based only on the augmented original design matrix, that is, on
$\bigl[{\X\atop \mathbf0}\bigr]$,
it does not depend on the observed response $\y$.

In Section~\ref{sec:exper}, we analyze real HIV data with one case of
the form
$p<n < 2p$ and, thereby, show that the basic knockoff method can be
adapted to situations in which $n>p$ even if $n\ngeq2p$, so that
it applies all the way to the limit of model identifiability.

\subsection{Symmetric statistics}\label{sec:Wstats}

We next consider a statistic $W ([\X\ \Xproxy],\y) \in\R^p$
with large positive values of $W_j$, giving evidence that $\beta_j\neq
0$, and introduce two simple properties.

\begin{definition}
The statistic $\mathbf W$ is said to obey the \textit{sufficiency
property} if $\mathbf W$ depends only on the Gram matrix and on
feature-response inner products; that is, we can write
\[
\mathbf W = f \bigl([\X\ \Xproxy]^\top[\X\ \Xproxy],[\X\ \Xproxy
]^\top \y \bigr)
\]
for some $f: S_{2p}^+ \times\R^{2p} \rightarrow\R^p$, where
$S_{2p}^+$ is the cone of $2p \times2p$ positive semidefinite
matrices. (We call this the sufficiency property
since under Gaussian noise, $\X^\top\y$ is a sufficient statistic for
$\b$.)
\end{definition}

\begin{definition}
The statistic $\mathbf W$ is said to obey the \textit{antisymmetry
property} if swapping $\X_j$ and $\Xproxy_j$ has the effect of
switching the sign of $W_j$; that is, for any $S\subseteq\{1,
\ldots, p\}$,
\[
W_j \bigl([\X\ \Xproxy]_{\swap(S)},\y\bigr) = W_j
\bigl([\X\ \Xproxy], \y\bigr) \cdot\cases{ +1, &\quad $j\notin S$,\vspace*{2pt}
\cr
-1,&\quad
$j\in S$.}
\]
\end{definition}

Here, we write $[\X\ \Xproxy]_{\swap(S)}$ to mean that the columns $\X_j$
and $\Xproxy_j$ have been swapped in the matrix $[\X\ \Xproxy]$, for
each $j\in S$.
Formally, if $\mathbf V \in\R^{n\times2p}$ with columns
$\mathbf V_j$, then for each $j=1,\ldots,p$,
\[
(\mathbf V_{\swap(S)})_j = \cases{ \mathbf V_j, &\quad
$j\notin S$,\vspace*{2pt}
\cr
\mathbf V_{j+p}, & \quad$j\in S$, }\qquad (\mathbf
V_{\swap(S)})_{j+p}=\cases{ \mathbf V_{j+p}, &\quad $j\notin S$,
\vspace*{2pt}
\cr
\mathbf V_j,&\quad $j\in S$.}
\]
%

The statistic $\mathbf W$ we examined in Section~\ref{sec:new},
given in
equation \eqref{eqn:W}, obeys these two properties. The reason why the
sufficiency property holds is that the Lasso~\eqref{eqn:lasso} is
equivalent to
\[
\mbox{minimize}\quad \tfrac{1}{2} \mathbf{b}^\top\X^\top\X
\mathbf{b} - \mathbf{b}^\top\X^\top\y+ \lambda\|\mathbf{b}
\|_1,
\]
and thus depends upon the problem data $(\X, \y)$ through $\X^\top
\X$
and $\X^\top\y$ only.\footnote{If we would like to include an
intercept term in our model, that is, $\y=\beta_0\mathbf1 + \X\b+
\e$,
then the Lasso also depends on $\X^\top\mathbf1$ and $\y^\top
\mathbf
1$. In this case, we can apply our method as long as the knockoffs
additionally satisfy $\Xproxy^\top\mathbf1 = \X^\top\mathbf1$.}
Note that
the antisymmetry property in \eqref{eqn:W} is explicit. This is only one
example of a statistic of this type; other examples include the following:
\begin{longlist}[(1)]
\item[(1)] $W_j = |\X_j^\top\y| - |\Xproxy_j^\top\y|$, which
simply compares
marginal correlations with the response.\vadjust{\goodbreak}

\item[(2)] $W_j =|\hat{\beta}{}^{\mathsf{LS}}_j|-|\hat{\beta
}{}^{\mathsf{LS}}_{j+p}|$ or
$W_j=|\hat{\beta}{}^{\mathsf{LS}}_j|^2-|\hat{\beta}{}^{\mathsf{LS}}_{j+p}|^2$,
where $\hat{\b}{}^{\mathsf{LS}}$ is the least-squares solution
obtained by regressing $\y$ on the augmented design,
$\hat{\b}{}^{\mathsf{LS}} = ([\X\ \Xproxy]^\top[\X\ \Xproxy]
)^{-1}[\X\ \Xproxy]^\top\y$.

\item[(3)] Define $Z_j$ as in Section~\ref{sec:new}, $Z_j=\sup\{
\lambda:\hat
{\beta}_j(\lambda)\neq0\}$
for $j=1,\ldots,2p$ where $\hat{\b}(\lambda)$ is the solution to the
augmented
Lasso model regressing $\y$ on $ [\X\ \Xproxy]$.
We may then take $W_j= (Z_j\vee Z_{j+p} ) \cdot
\sign(Z_j-Z_{j+p})$,
but may also consider other options such as $W_j=Z_j - Z_{j+p}$,
or alternately we can take $W_j=|\hat{\beta}_j(\lambda)| - |\hat
{\beta}_{j+p}(\lambda)|$
for some fixed value of $\lambda$.
(To maintain consistency with the rest of this section,
the notation here is slightly different than in Section~\ref{sec:new}, with
$Z_{j+p}$ instead of $\tilde{Z}_j$ giving the $\lambda$ value when
$\Xproxy_j$ entered
the Lasso model.)

\item[(4)] The example above, of course, extends to all penalized likelihood
estimation procedures of the form
\[
\mbox{minimize}\qquad \tfrac{1}{2} \|\y- \X\mathbf{b}\|_2^2
+ \lambda P(\mathbf{b}),
\]
where $P(\cdot)$ is a penalty function, the Lasso being only one such
example. We can again define $\mathbf W$ by finding the $\lambda$ values
at which each feature enters the model, or by fixing $\lambda$ and comparing
coefficients in $\hat{\beta}$. In particular, we may consider methods that
use a nonconvex penalty, such as SCAD \cite{fan2001variable}, where
the nonconvexity of the penalty $P(\mathbf{b})$ reduces bias in estimating
the vector of coefficients.

\item[(5)] We can also consider a forward selection procedure
\cite{draper1981applied}: initializing the residual as $\mathbf r_0=\y$,
we iteratively choose variables via
\[
j_t = \arg\max_j \bigl|\langle{
\X_j},{\mathbf r_{t-1}}\rangle\bigr|
\]
and then update the
residual $\mathbf r_t$ by either regressing the previous residual
$\mathbf
r_{t-1}$ on $\X_{j_t}$ and taking the remainder, or alternately
using orthogonal matching pursuit~\cite{pati1993orthogonal}, where
after selecting $j_t$ we define $\mathbf r_t$ to be the residual of the
least square regression of $\y$ onto
$\{\X_{j_1},\ldots,\X_{j_t}\}$. As before, however, we apply this
procedure to the augmented design matrix $[\X\ \Xproxy]$. Next let
$Z_1,\ldots,Z_{2p}$ give the reverse order in which the $2p$
variables (the originals and the knockoffs) entered the model; that is,
$Z_j=2p$ if $\X_j$ entered first; $Z_j=2p-1$ if $\X_j$ entered
second, etc. The statistics $W_j= (Z_j\vee Z_{j+p} ) \cdot
\sign(Z_j-Z_{j+p})$ then reflect the time at which the original
variable $\X_j$ and the knockoff variable $\Xproxy_j$ entered the
model.

\item[(6)] Generalizing the forward selection procedure, we can consider
algorithms producing solution ``paths'' $\hat{\b}^{\lambda}$, where
the path
depends on the data only through $\X^\top\X$ and $\X^\top\y$.
Examples of such methods include LARS \cite{efron2004least}, the least
angle regression
method (closely related to the LASSO)
and MC$+$ \cite{zhang2010nearly}, a nonconvex method aimed at reducing
bias in the estimated
coefficients. Applying any such
method to the augmented design matrix $[\X\ \Xproxy]$, we would then
extract the order in which features enter the path, to determine the
values $Z_j$
and $Z_{j+p}$.
\end{longlist}
Clearly, the possibilities are endless.

\subsection{Exchangeability results}

Despite the fact that the statistics $W_j$, $j \in\{1, \ldots, p\}$,
are dependent and have marginal distributions that are complicated
functions of the unknown parameter vector $\bolds\beta$, our selection
procedure provably controls the false discovery rate, as stated in
Theorems \ref{thm:knockoff} and \ref{thm:knockoffplus} from
Section~\ref{sec:new}.
In this section, we establish a property of the statistics $W_j$
that we will use to prove our main results on FDR control.

In fact, the construction of the knockoff features and the symmetry of
the test
statistic are in place to achieve a crucial property, namely, that the
signs of the $W_j$'s are i.i.d. random for the ``null hypotheses'' and
furthermore
are independent from the magnitudes $|W_j|$ for all $j$, and from
$\sign(W_j)$ for the ``nonnull
hypotheses'' $j$.

\begin{lemma}[(i.i.d. signs for the nulls)]
\label{lem:key}
Let $\bolds\varepsilon\in\{\pm1\}^p$ be a sign sequence independent of
$\mathbf W$, with $\varepsilon_j=+1$ for all nonnull $j$ and
$\varepsilon_j\iidsim\{\pm1\}$ for null $j$. Then
\[
(W_1, \ldots, W_p) \eqd(W_1 \cdot
\varepsilon_1, \ldots, W_p \cdot\varepsilon_p).
\]
\end{lemma}

This property fully justifies our earlier statement \eqref{eqn:eqd} that
$\#\{j: \beta_j=0, W_j \le-t\}$ has the same distribution as $\#\{
j:\beta_j=0,
W_j \ge t\}$. Indeed, conditional on $|\mathbf W| = (|W_1|, \ldots,
|W_p|)$, both these random variables follow the same binomial
distribution, which implies that their marginal distributions are
identical. In turn, this gives that $\widehat{\mathsf{FDP}}(t)$ from
Section~\ref{sec:new} is an estimate of the true false discovery proportion
$\mathsf{FDP}(t)$.

The i.i.d. sign property for the nulls is a consequence of the two following
exchangeability properties for $\X$ and $\Xproxy$:

\begin{lemma}[(Pairwise exchangeability for the features)]\label{lem:pairs_cov}
For any subset $S\subset\{1, \ldots, p\}$,
\[
[\X\ \Xproxy]_{\swap(S)}^\top[\X\ \Xproxy]_{\swap(S)} = [\X\ \Xproxy]^\top[\X\ \Xproxy].
\]
That is, the Gram matrix of $[\X\ \Xproxy]$ is unchanged when we swap
$\X
_j$ and
$\Xproxy_j$ for each $j\in S$.
\end{lemma}

\begin{pf} This follows trivially from the definition of $\mathbf G =
[\X\ \Xproxy]^\top[\X\ \Xproxy]$ in~\eqref{eqn:proxy2}.
\end{pf}

\begin{lemma}[(Pairwise exchangeability for the
response)]\label{lem:pairs}
For any subset $S$ of nulls,
\[
[\X\ \Xproxy]_{\swap(S)}^\top\y\eqd[\X\ \Xproxy]^\top\y.
\]
That is, the distribution of the product $[\X\ \Xproxy]^\top\y$ is
unchanged when we swap $\X_j$ and
$\Xproxy_j$ for each $j\in S$, as long as none of the swapped features
appear in the true model.
\end{lemma}

\begin{pf}
Since $\y\sim\mathcal{N}(\X\b,\sig^2\ident)$, for any $S'$, we have
\[
[\X\ \Xproxy]_{\swap(S')}^\top\y\sim N \bigl([\X\ \Xproxy]_{\swap
(S')}^\top
\X\b,\sig^2 [\X\ \Xproxy]_{\swap(S')}^\top[\X\ \Xproxy]
_{\swap(S')} \bigr).
\]
Next we check that the mean and variance calculated here are the
same for $S'=S$ and for $S'=\varnothing$. Lemma~\ref{lem:pairs_cov} proves
that the variances are equal. For the means, since $\X_j^\top
\X_i=\Xproxy_j^\top\X_i$ for all $i\neq j$, and
$\operatorname{support}(\b)\cap S=\varnothing$, we see that $\X_j^\top\X
\b
=\Xproxy_j^\top\X\b$ for all $j\in S$, which is sufficient.
\end{pf}

\begin{pf*}{Proof of Lemma~\ref{lem:key}}
For any set $S\subset\{1, \ldots, p\}$, let $\mathbf{W}_{\swap(S)}$ be
the statistic we would get if we had replaced $[\X\ \Xproxy]$ with
$[\X\ \Xproxy]_{\swap(S)}$ when calculating $\mathbf W$. The
anti-symmetry property gives
\[
\mathbf{W}_{\swap(S)} = (W_1 \cdot\varepsilon_1,
\ldots, W_p \cdot\varepsilon_p), \qquad\varepsilon_j
= \cases{ +1, & \quad$j \notin S$,\vspace*{2pt}
\cr
-1, &\quad $j \in S$.}
\]
Now let $\bolds\varepsilon$ be as in the statement of the lemma, and let
$S=\{j:\varepsilon_j=-1\}$. Since $S$ contains only nulls, Lemmas
\ref{lem:pairs_cov} and \ref{lem:pairs} give
\begin{eqnarray*}
\mathbf{W}_{\swap(S)}&=&f \bigl([\X\ \Xproxy]_{\swap(S)}^\top(\X\ \Xproxy)_{\swap(S)},[\X\ \Xproxy]_{\swap(S)}^\top\y \bigr)
\\
&\eqd &f \bigl([\X\ \Xproxy]^\top[\X\ \Xproxy],[\X\ \Xproxy]^\top\y
\bigr) = \mathbf W.
\end{eqnarray*}
This proves the claim.
\end{pf*}

\subsection{Proof sketch for main results}
\label{sec:sketch}

With the exchangeability property of the $W_j$'s in place, we sketch
the main ideas behind the proof of our main results,
Theorems \ref{thm:knockoff} and \ref{thm:knockoffplus}.
The full details will be presented later, in Sections \ref
{sec:sequential} and
in the supplementary materials \cite{knockoffsupplement},
where we will see that our methods can be
framed as special cases of a sequential hypothesis testing procedure.
Such sequential procedures are not specifically about the regression
problem we consider here, and this is the reason why we prefer
postponing their description as not to distract from the problem at
hand.

We restrict our attention to the knockoff$+$ method for simplicity. To
understand how knockoff$+$ controls FDR, we consider step 3 of the
method, where after calculating the statistics $W_j$, we choose the
data-dependent threshold $\thr$ given by \eqref{eqn:threshplus}.
By definition, the FDP is equal to
%
\begin{eqnarray} \label
{eqn:FDPbound}
\mathsf{FDP} &=& \frac{\#\{j:\beta_j=0\mbox{ and
}W_j\geq T\}}{\#\{j:W_j\geq T\}\vee1}
\nonumber
\\
& \leq&\frac{1+\#\{j:W_j\leq-T\}}{\#\{j:W_j\geq T\}\vee1} \cdot \frac{\#\{j:\beta_j=0\mbox{ and }W_j\geq T\}}{1+\#\{j:\beta_j=0
\mbox{ and }W_j\leq-T\}}
\\
&\leq& q \cdot\frac{\#\{j:\beta_j=0\mbox{ and
}W_j\geq T\}}{1+\#\{j:\beta_j=0 \mbox{ and }W_j\leq-T\}};\nonumber
\end{eqnarray}
the first inequality follows from the fact that $\#\{j:\beta_j=0
\mbox{ and }W_j\leq-T\} \le\#\{j: W_j\leq-T\}$ and the second from
the definition of $T$.
Since $T$ is the first time a ratio falls below $q$, it turns out that
we may view $T$ as a stopping time. In fact, the main step of our
proof is to show that $T$ is a stopping time for the supermartingale
$V^+(T)/(1+V^-(T))$, where $V^\pm(t) = \# \{j \mbox{ null } : |W_j|
\ge t \mbox{ and } \mathsf{sign}(W_j) = \pm1\}$; the details are deferred
to the supplementary materials \cite{knockoffsupplement}.
By the optional stopping time theorem,
\[
\mathbb{E}\biggl[{\frac{V^+(T)}{1+V^-(T)}}\biggr] \le\mathbb{E}\biggl[{
\frac
{V^+(0)}{1+V^-(0)}}\biggr] = \mathbb{E}\biggl[{\frac{V^+(0)}{1+p_0 - V^+(0)}}\biggr] \le1,
\]
%
where the last step comes from a property of the binomial distribution
proved in the supplementary materials \cite{knockoffsupplement}; note
that since
$\mathsf{sign}(W_j)\iidsim\{\pm1\}$ for the null features $j$, then
$V^+(0)$ is distributed as a $\mathsf{Binomial}(p_0,1/2)$ random
variable. (For the purposes of this proof sketch, we assume here that
$W_j\neq0$ for all $j$ for simplicity.) This, together with
\eqref{eqn:FDPbound}, proves FDR control.
The proof for the knockoff method is similar, and we refer the
reader to Section~\ref{sec:sequential} and to the supplementary
materials \cite
{knockoffsupplement}
for details.

\section{Comparison with other variable selection techniques}
\label{sec:compare}

There are of course many other variable selection techniques, based on
ideas from Benjamini and Hochberg or perhaps based on permuted designs
rather than knockoffs, which may be designed with the goal of keeping
FDR under control. In this section, we review some of these procedures
and compare some of them empirically.

\subsection{Comparing to a permutation method}\label{sec:permutation}

To better understand the ideas behind our method, we next ask whether
we could have constructed the matrix of knockoff features $\Xproxy$
with a simple permutation. Specifically, would the above results hold
if instead of constructing $\Xproxy$ as above, we use a matrix
$\X^{\pi}$, with entries given by
\[
\X^{\pi}_{i,j}=\X_{\pi(i),j}
\]
for some randomly chosen permutation
$\pi$ of the sample indices $\{1,\ldots,n\}$? In particular, the matrix
$\X^{\pi}$ will always satisfy
$
\X^{\pi}{}^\top\X^{\pi}= \X^\top\X$,
and so the permuted covariates display the same correlation structure
as the original covariates, while breaking association with the
response $\y$ due to the permutation.

Permutation methods are widely used in applied research. While they
may be quite effective under a global null, they may fail to yield
correct answers in cases other than the global null; see also
\cite{chung2013,romano2013} for other sources of problems associated
with permutation methods. Consequently, inference in practical
settings, where some signals do exist, can be quite distorted. In the
linear regression problem considered here, a permutation-based
construction can dramatically underestimate the FDP in cases where
$\X$ displays nonvanishing correlations. To understand why, suppose
that the features $\X_j$ are centered. Then $\X^\top\X=\X^{\pi}
{}^\top\X^{\pi}=\bolds\Sigma$,
but $\X^\top\X^{\pi}\approx\mathbf0$.
In particular,
the exchangeability results (Lemmas \ref{lem:pairs_cov} and \ref{lem:pairs})
will not hold for the augmented matrix $[\X\ \X^{\pi}]$, and this can lead
to extremely poor control of FDR.

To see this empirically, consider a setting with positive correlation between
features. We generate each row of $\X\in\R^{300\times100}$
i.i.d. from a $\mathcal{N}(\mathbf0,\bolds\Theta)$ distribution, where
$\Theta_{ii}=1$ for all $i$ and $\Theta_{ij}=0.3$ for all $i\neq
j$. We then center and normalize the columns of $\X$ and define
%
\begin{equation}
\label{eqn:perm_example}\y= 3.5\cdot(\X_1 + \cdots+ \X_{30}) + \e\qquad
\mbox{where }\e\sim\mathcal{N}(\mathbf{0}, \ident_n).
\end{equation}
Next we fit the Lasso path \eqref{eqn:lasso} for the response $\y$ and
the augmented
design matrix $[\X\ \X^{\pi}]$. Figure~\ref{fig:perm_fail} shows that while
many of the original null features enter the model at moderate
values of $\lambda$, the permuted features do not enter the Lasso path until
$\lambda$ is extremely small; the difference arises from the fact that
only the original
null features are correlated with the signals $\X_1,\ldots,\X_{30}$.
In other words, the $\X^{\pi}_j$'s are not good knockoffs of the $\X_j$'s
for the nulls $j=31,\ldots,100$---they behave very differently in the
Lasso regression.

Next we test the effect of these issues on FDR control.
Using the permuted features $\X^{\pi}$ in place of $\Xproxy$, we
proceed exactly as for the knockoff method
[see \eqref{eqn:W} and Definition~\ref{def:knockoff}] to select a
model. We compare
to the knockoff method and obtain
the following FDR, when the target
FDR is set at $q=20\%$:\vspace*{12pt}

{\fontsize{9}{11}\selectfont
\begin{center}
\begin{tabular*}{220pt}{@{\extracolsep{\fill}}lc@{}}
\hline
&\textbf{FDR over 1000 trials}\\
& \textbf{(nominal level $\bolds{q=20\%}$)}\\\hline
Knockoff method&12.29\%\\
Permutation method&45.61\%\\
\hline
\end{tabular*}\vspace*{12pt}
\end{center}}
\noindent We note that there are many
ways in which the permuted features $\X^{\pi}$ may be used to try to
estimate or control FDR,
but in general such methods will suffer from similar issues arising
from the lack of correlation
between the permuted and the original features.

\begin{figure}

\includegraphics{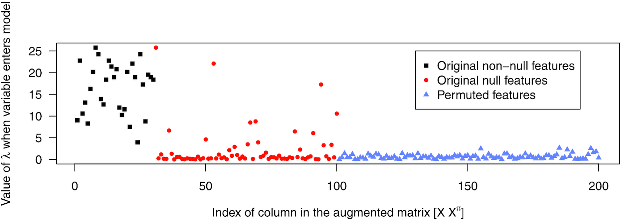}

\caption{Results of the Lasso path, with simulated data
specified in \protect\eqref{eqn:perm_example}. Many
of the null features $\X_j$ for $j=31,\ldots,100$ enter the Lasso model
earlier (i.e., at higher values of $\lambda$) than most of the
permuted features,
leading to loss of FDR control.}
\label{fig:perm_fail}
\end{figure}

\subsection{The Benjamini--Hochberg procedure and variants}\label{sec:BH}

The Benjamini--Hochberg (BHq) procedure \cite{BH95} is a hypothesis
testing method known to control FDR under independence. Given z-scores
$Z_1,\ldots,Z_p$ corresponding to $p$ hypotheses being tested so that
$Z_j \sim\mathcal{N}(0,1)$ if the $j$th hypothesis is null, the
procedure\footnote{We present BHq in the notation from
\cite{Storey02} and use z-scores rather than $p$-values to
facilitate comparison with our methods. Section~4 of \cite{Storey02} proves
that this procedure is equivalent to BHq.} rejects a hypothesis
whenever $|Z_j|\geq\thr$, where $\thr$ is a data-dependent threshold
given by
%
\begin{eqnarray}
\label{eqn:Bhq_intro} \thr=\operatorname{min} \biggl\{t:\frac{p \cdot
\mathbb{P}\{{|\mathcal{N}(0,1)| \geq t}\}}{ \#\{j: |Z_j| \geq t\}
}\leq q \biggr\}
\nonumber
\\[-8pt]
\\[-8pt]
\eqntext{\mbox{(or $\thr=+\infty$ if this set is empty)},}
\end{eqnarray}
for a desired FDR
level $q$. Note that for any $t$, the number of null hypotheses with
$|Z_j|\geq t$ can be estimated by $\pi_0p\cdot
\mathbb{P}\{{|\mathcal{N}(0,1)|\geq t}\}$, where $\pi_0p$ is the
total number of
null hypotheses. For $\pi_0<1$, then the fraction in the definition of
\eqref{eqn:Bhq_intro} is an overestimate of the FDP by a factor of
$(\pi_0)^{-1}$; see \cite{EfronBook} and references therein.

Turning to the problem of variable selection in regression, the BHq
procedure may be applied by calculating the least-squares estimate,
\[
\hat{\b}{}^{\mathsf{LS}}= \bigl(\X^\top\X \bigr)^{-1}
\X^\top\y.
\]
For Gaussian noise as in \eqref{eqn:model}, these fitted coefficients
follow a $\mathcal{N}(\b,\sigma^2 \bolds\Sigma^{-1})$ distribution,
where we
recall that $\bolds\Sigma= \X^\top\X$. Therefore, setting
$Z_j = \hat{\beta}{}^{\mathsf{LS}}_j/\break \sigma\sqrt{(\bolds\Sigma^{-1})_{jj}}$
yields z-scores, that is, marginally $Z_j\sim\mathcal{N}(0,1)$ whenever
$\beta_j=0$. Variables are then selected using the data-dependent
threshold given in \eqref{eqn:Bhq_intro}.

Under orthogonal designs in which $\X^\top\X$ is a diagonal matrix,
the $Z_j$'s are independent; in this setting, Benjamini and Hochberg
\cite{BH95} prove that the BHq procedure controls FDR at the level
$\pi_0\cdot q$; see Section~\ref{sec:BHq_ortho} for a comparison of
knockoff methods
with BHq in the orthogonal design setting.
Without the assumption of independence, however,
there is no such guarantee. In fact, it is not hard to construct
designs with only two variables such that the BHq procedure does not
control the FDR at level $q$. FDR control has been established for
test statistics obeying the positive regression dependence on a subset
property (PRDS) introduced in \cite{BY01}. The problem is that the
PRDS property does not hold in our setting, for two reasons. First,
for one-sided tests
where the alternative is $\beta_j > 0$, say, one would reject for
large values of $Z_j$. Now for the PRDS property to hold we would need
to have $(\bolds\Sigma^{-1})_{ij} \ge0$ for all nulls $i$ and all $j$,
which rarely is in effect. Second, since the signs of the coefficients
$\beta_j$
are in general unknown, we are performing two-sided tests where we
reject for large values of $|Z_j|$ rather than $Z_j$; these
absolute-value statistics
are not known to be PRDS either, even under positive values of $(\bolds
\Sigma^{-1})_{ij}$.\vadjust{\goodbreak}

Against this background, Benjamini and Yekutieli \cite{BY01} show
that the BHq procedure yields FDR bounded by $\pi_0q\cdot S(p)$
regardless of the dependence among the z-scores, where
$S(p)=1+1/2+\cdots+1/p\approx\log p+0.577$.
Therefore, if we define $\thr$ as in \eqref{eqn:Bhq_intro} but with
$q/S(p)$ in place of $q$, then we are again guaranteed a bound on FDR.

Finally, as another option, we can ``whiten the noise'' in $\hat{\b}$
before applying BHq. Specifically, let $\mathbf Z' \sim\mathcal
{N}(0,\sigma^2\cdot
(\lambda_0^{-1}\ident- \bolds\Sigma^{-1}))$ be drawn independently from
the data, where $\lambda_0=\lambda_{\operatorname{min}}(\bolds\Sigma)$. Then
%
\begin{equation}
\label{eqn:whiten}\hat{\b}+Z' \sim\mathcal{N} \bigl(\b,
\sigma^2 \lambda_0^{-1} \ident \bigr),
\end{equation}
and we can then apply BHq to the z-scores given by $Z_j = (\hat{\b
}_j+Z'_j)/\sigma\sqrt{\lambda_0}$.
Since
these z-scores are now independent, applying BHq yields an FDR of at
most $\pi_0 q$.

For all of the variants of BHq considered here, FDR control is
estimated or guaranteed to be at a level of $\pi_0q $, which is lower
than the nominal level $q$; that is, the method is more conservative than
desired. However, here we are primarily interested in a sparse setting
where $\pi_0\approx1$, and so this will not have a strong effect.

\subsection{Empirical comparisons with the Benjamini--Hochberg method
and variants}\label{sec:simulations}

We now test our method\footnote{Code for the knockoff method
is available via the \textsc{R} package \textsc{knockoff},
at \url{http://cran.r-project.org/web/packages/knockoff/},
and the \textsc{Matlab} package \textsc{knockoff\_matlab},
at \url{http://web.stanford.edu/\textasciitilde
candes/Knockoffs/package_matlab.html}.}
in a range of settings, comparing it to BHq and
its variants, and examining the effects of sparsity level, signal
magnitude and feature correlation.

\subsubsection{Comparing methods}\label{sec:setting1}
We begin with a comparison of seven methods: the equi-variant and the
SDP constructions for the knockoff and\break knockoff$+$ filters, the BHq
procedure, the
BHq procedure with the log-factor correction \cite{BY01} to guarantee
FDR control with dependent z-scores [i.e., this applies BHq with
$q/S(p)$ replacing $q$] and the BHq procedure
with whitened noise, as in \eqref{eqn:whiten}. To summarize our earlier
discussions, the equi-variant and SDP constructions for knockoff$+$, the
BHq method with the log-factor correction, and the BHq method with
whitened noise are all guaranteed to control FDR at the nominal level
$q$; the other methods do not offer this exact guarantee.

\begin{table}
\tabcolsep=0pt
\caption{FDR and power in the setting of Section \protect\ref{sec:setting1}
with $n = 3000$ observations, $p = 1000$ variables and $k = 30$
variables in the model with regression coefficients of magnitude 3.5.
Bold face font highlights those methods that are known theoretically
to control FDR at the nominal level $q = 20\%$}
\label{tab:setting1}
\begin{tabular*}{\textwidth}{@{\extracolsep{\fill}}lccc@{}}
\hline
&\textbf{FDR (\%)}&&\textbf{Theoretical}\\
&\textbf{(nominal level}&&\textbf{guarantee}\\
\textbf{Method}& \multicolumn{1}{c}{$\bolds{q=20\%}$\textbf{)}}&\textbf{Power (\%)}&\textbf{of FDR control?}\\
\hline
\textbf{Knockoff$\bolds{+}$ (equivariant construction)}&\textbf{14.40}&\textbf{60.99}&
\textbf{Yes}\\
Knockoff (equivariant construction) & 17.82 &66.73&No\\
\textbf{Knockoff$\bolds{+}$ (SDP construction)} &\textbf{15.05}&\textbf{61.54}&\textbf{Yes}\\
Knockoff (SDP construction) & 18.72 &67.50&No\\[3pt]
Benjamini--Hochberg (BHq) \cite{BH95} & 18.70 &48.88&No\\
\textbf{BHq $\bolds{+}$ log-factor correction \cite{BY01}} & \phantom{0}\textbf{2.20}
&\textbf{19.09}&\textbf{Yes}\\
\textbf{BHq with whitened noise} &\textbf{18.79}& \phantom{0}\textbf{2.33}&\textbf{Yes}\\
\hline
\end{tabular*}
\end{table}

In this simulation, we use the problem size $n=3000$, $p=1000$ and a
number $k = 30$ of variables in the model. We first draw
$\X\in\R^{n\times p}$ with i.i.d. $\mathcal{N}(0,1)$ entries, then
normalize its
columns. Next, to define $\b$, we choose $k=30$ coefficients at
random and choose $\beta_j$ randomly from $\{\pm A\}$ for each of the
$k$ selected coefficients, where $A=3.5$ is the signal
amplitude. Finally, we draw $\y\sim\mathcal{N}(\X\b, \ident)$.
The signal
amplitude $A=3.5$ is selected because $3.5$ is approximately the
expected value of $\max_{1 \le j \le p} |\X_j^\top\mathbf z|$ where
$\mathbf
z \sim\mathcal{N}(\mathbf0,\ident)$ (each $\X_j^\top\y$ is
approximately a
standard normal variable if $\beta_j=0$). Setting the signal
amplitude to be near this maximal noise level ensures a setting where
it is possible, but not trivial, to distinguish signal from noise.

Table~\ref{tab:setting1} displays the resulting FDR and power obtained by
each method, averaged over $600$ trials.
Empirically, all of the methods result in an FDR that is near or below
the nominal level $q=20\%$. Comparing
their power, knockoff and knockoff$+$ (power $>60\%$) significantly
outperform BHq (power $\approx49\%$).

Comparing the equi-variant and SDP constructions for the knockoff and
knockoff$+$ methods, the SDP construction achieves slightly higher power
for both knockoff and knockoff$+$. Finally, the two variants of BHq
considered do offer theoretical control of FDR, but empirically
achieve very poor power in this simulation. From this point on,
we thus restrict our attention to three methods: the knockoff method using
the SDP construction given in \eqref{eqn:set_s_sum}, the knockoff$+$ method
with the same SDP construction and BHq.

\subsubsection{Effect of sparsity level, signal amplitude and feature
correlation}\label{sec:setting2}

\begin{figure}

\includegraphics{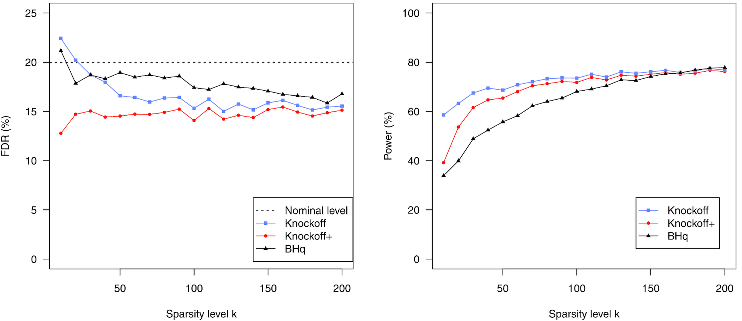}

\caption{Testing the knockoff, knockoff$+$, and BHq methods at nominal
level $q = 20\%$ with varying sparsity level $k$. Here $n=3000$,
$p=1000$ and $A=3.5$, and the figures show mean FDR and mean power
averaged over $600$ trials.}
\label{fig:vary_k}
\end{figure}

\begin{figure}[b]

\includegraphics{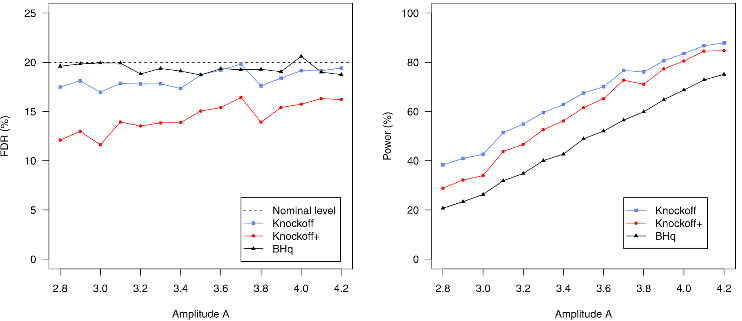}

\caption{Testing the knockoff, knockoff$+$ and BHq methods at nominal
level $q = 20\%$ with varying signal amplitudes $A$. Here $n=3000$,
$p=1000$ and $k=30$, and the figures show mean FDR and mean power
averaged over $200$ trials.}
\label{fig:vary_amp}
\end{figure}

\begin{figure}

\includegraphics{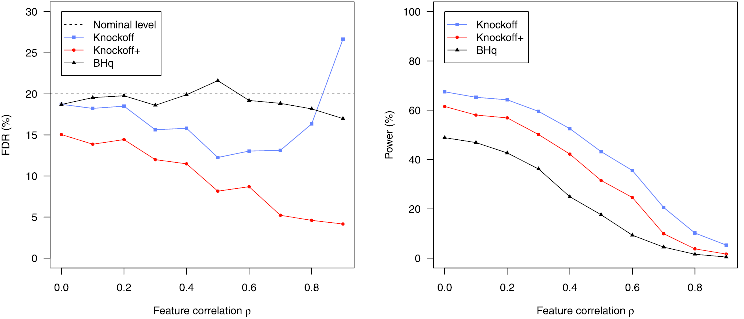}

\caption{Testing the knockoff, knockoff$+$ and BHq methods at nominal
level $q = 20\%$ with varying feature correlation levels. The
correlation parameter $\rho$ controls the tapered correlation
structure of the design matrix, where columns $\X_j$ and $\X_k$ are
generated from a distribution with correlation $\rho^{|j-k|}$. Here
$n=3000$, $p=1000$, $k=30$, $A=3.5$ and the figures show mean FDR
and mean power averaged over $200$ trials.}
\label{fig:vary_rho}
\end{figure}

Next, we consider the effect of varying the sparsity level $k$, the
signal amplitude $A$
or the feature correlation level, when comparing the performance for
the knockoff, knockoff$+$,
and BHq methods. We test each of these questions separately as follows:
\begin{itemize}
\item Effect of the sparsity level $k$: we test values $k= 10, 20,
30,\ldots, 200$
while fixing signal amplitude $A=3.5$ (all other settings are identical
to Section~\ref{sec:setting1}).
\item Effect of signal amplitude $A$: we test values $A=2.8,2.9,\ldots
,4.2$ while
fixing sparsity level $k=30$ (all other settings are identical to
Section~\ref{sec:setting1}).
\item Effect of feature correlation: we
generate the rows of $\X$ from a $\mathcal{N}(0,\bolds\Theta_{\rho})$
distribution, where $(\bolds\Theta_{\rho})_{jk}=\rho^{|j-k|}$,
for correlation level $\rho=0,0.1,\ldots,0.9$. (In the case that
$\rho=0$, we simply set $\bolds\Theta=\ident$, as before.)
We then normalize the columns of $\X$ and generate $\b$ and $\y$
in the same manner as in Section~\ref{sec:setting1} with sparsity
level $k=30$
and signal amplitude $A=3.5$.
\end{itemize}
The mean FDR and mean power over
$200$ trials are displayed in Figures \ref{fig:vary_k}, \ref{fig:vary_amp}
and~\ref{fig:vary_rho}, respectively.

Examining these results, we see that across the three experiments, all
three methods
successfully control FDR at the nominal level $q=20\%$, with one
notable exception: for the correlated
design, the knockoff method controls FDR for $\rho\leq0.8$, but
shows a higher FDR level of $26.67\%$ when $\rho=0.9$. This is
consistent with our theoretical result, Theorem~\ref{thm:knockoff},
which guarantees that the knockoff method controls a modified form of
the FDR that is very similar to the FDR when a high number
of variables are selected, but may be quite different when a small
number of variables is selected. At the
higher values of $\rho$, we make so few discoveries (typically less
than $5$ for both knockoff and knockoff$+$) that the additional
``$+1$'' appearing in the knockoff$+$ method makes a substantial difference.

Turning to the power of the three methods, we see that both knockoff and
knockoff$+$ offer as much or more power than BHq across all settings,
with a strong advantage
over BHq at low and moderate values
of $k$ across the range of signal amplitude levels and correlation
levels---these
methods successfully leverage the sparse structure of
the model in this high-dimensional setting. In the study of sparsity
level (Figure~\ref{fig:vary_k}),
for higher values of $k$, when the
problem is no longer extremely sparse, the power of BHq catches up
with the knockoff and knockoff$+$ methods. As expected, each method shows
lower power
with high correlations, reflecting the difficulty
of telling apart neighboring features that are strongly correlated, and
at low signal amplitude
levels.

In summary, we see that the knockoff and knockoff$+$ methods have higher
power than
BHq while having a lower type I error. In the language of multiple
testing, this says that these methods detect more true effects while
keeping the
fraction of false discoveries at a lower level, which makes findings
somehow more reproducible.

\subsection{Relationship with the Benjamini--Hochberg procedure under
orthogonal designs}
\label{sec:BHq_ortho}
%

\begin{figure}
\centering
\begin{tabular}{@{}cc@{}}

\includegraphics{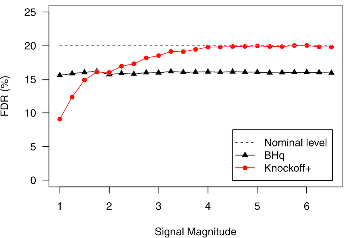}
 & \includegraphics{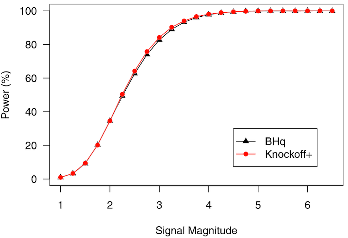}\\
\footnotesize{(a) FDR} & \footnotesize{(b) Power}
\end{tabular}
\caption{FDR and power of the BHq and knockoff$+$ methods, plotted
against the size $A$ of the regression coefficients (signal
magnitude), averaged over $1000$ trials. The nominal FDR level $q$
is set to 20\%. The $2000 \times1000$ design $\X$ is orthogonal,
and the number of true signals is 200 so that the fraction of nulls
is $\pi_0 = 0.8$.}\label{fig:orth}
\end{figure}
%
%
%
%

The Benjamini--Hochberg (BHq) procedure is known to control FDR in the
setting where the statistics for the null hypotheses are mutually
independent. In the regression setting where our statistics arise from
the least-squares coefficients $\hat{\b}\sim
\mathcal{N}(\b,\sigma^2(\X^\top\X)^{-1})$ (as in Section~\ref
{sec:BH}), the
coefficients of $\hat{\b}$ are mutually independent if and only if
$\X^\top\X$ is a diagonal matrix---the orthogonal design setting.
In this section, we consider an orthogonal design and compare the
knockoff filter and BHq side-by-side to understand the similarities and
differences in how these methods work. Figure~\ref{fig:orth}
demonstrates this
comparison empirically over a range of signal amplitude levels. Note
the contrasting FDR behavior between the two methods, even though the
power is essentially identical in each setting. In particular, we see
the following:
\begin{longlist}[(1)]
\item[(1)] The two methods both control FDR (as guaranteed by the theory)
and achieve nearly identical power over a range of signal
amplitudes.
\item[(2)] Theoretically and empirically, regardless of signal amplitude,
the FDR of BHq is given by $\pi_0 q$, where $q=20\%$ is the nominal
FDR level and $\pi_0$ is the proportion of null hypotheses,
$\pi_0=\frac{p-k}{p}$, as is shown in \cite{BH95}.
\item[(3)] In contrast, the FDR of the knockoff method varies over the
range of signal amplitudes. When the signal amplitude is high
enough for the power to be substantial, the FDR of the knockoff
method approaches $q$, rather than $\pi_0 q$; that is, the knockoff
method is implicitly correcting for the proportion of nulls, and
achieving the target FDR level. When the signal amplitude is so low
that power is near zero, the knockoff method has an extremely low
FDR (far lower than the nominal level~$q$), which is desirable in a
regime where we have little chance of finding the true signals.
\end{longlist}
%
A theoretical explanation of
these observations is given in the supplementary materials \cite
{knockoffsupplement}.

\subsection{Other methods}
Finally, we briefly mention several other approaches that are related
to the goal of this work. First, we discuss two methods presented
Miller \cite{MillerBook,Miller84} to control false positives in
forward selection
procedures for linear regression. The first method
creates ``dummy'' variables whose entries are drawn
i.i.d. at random. The forward selection procedure is then applied to
the augmented list of variables and is run until the first time it selects
a dummy variable. This approach is similar
in flavor to our method, but the construction of the dummy
variables does not account for correlation among the existing
features and therefore may lose FDR control in a correlated
setting. 
Miller \cite{Miller84} also proposes a second method, which makes use
of a key observation that we also use extensively in our work: after
selecting $m$ variables, if all of the
true features have already been selected, then the remaining residual is
simply Gaussian noise and is therefore
rotationally invariant. To test whether the next feature
should be included, \cite{Miller84} thus
compares to the null distribution obtained by applying random
rotations. 
In practice, true features and null features are nearly always
interspersed in the forward selection steps, and so this type of method
will not achieve exact control of the FDR for this reason.

We next turn to recent work by G'Sell et al. \cite{g2013false}, which
also gives an FDR-controlling procedure for the Lasso, without
constructing additional variables. This work uses the results from
\cite{lockhart2012significance,taylor2014post} that study the distribution
of the sequence of $\lambda$ values where new null variables enter the
model after all of the true signals have already been included in the
model. Consequently, the sequential testing procedure of
\cite{g2013false} controls FDR under an important assumption: the true
features must all enter the model before any of the null
features. Therefore, this method faces the same difficulty as the
second method of Miller \cite{Miller84} discussed above; when signals
and null features are interspersed along the Lasso path (as is
generally the case in practice, even when the nonzero regression
coefficients are quite large), this assumption is no longer satisfied,
leading to some increase of the FDR.

Next, the stability selection approach
\cite{larrywasserman2010stability,meinshausen2010stability} controls
false variable selection in the Lasso by refitting the Lasso model
repeatedly for subsamples of the data, and then keeps only those
variables that appear consistently in the resulting Lasso models.
These methods control false discoveries effectively
in practice and give theoretical guarantees of asymptotically consistent
model selection. For a finite-sample setting, however, there is no known
concrete theoretical guarantee for controlling false discoveries (with the
exception of a special case treated in Theorem~1 in \cite
{meinshausen2010stability}, which
for a linear model, reduces to the equi-variant
setting, $\Sigma_{ij}=\rho$ for all $i\neq j$).
Furthermore, these methods require computing the path of
Lasso models for many subsampled regressions containing $p$ candidate
variables each; in contrast, our method requires only a single
computation of the Lasso path, although for a model with $2p$
variables.

Finally, recent work \cite
{javanmard2013confidence,berk2012valid,van2013asymptotically,zhang2014confidence,voorman2014inference}
extends the classical
notions of confidence intervals and
$p$-values 
into the high-dimensional setting ($p\gg n$); although the
coefficients of the linear model are no longer identifiable in the
classical sense, these works perform inference under various
assumptions about the design and sparsity in the model.

\section{Experiment on real data: HIV drug resistance}\label{sec:exper}

We apply the knockoff filter to the task of detecting mutations in the
Human Immunodeficiency Virus Type 1 (HIV-1) that are associated with
drug resistance.\footnote{Data available online at
\url
{http://hivdb.stanford.edu/pages/published_analysis/genophenoPNAS2006/}. Code
for reproducing the analysis and figures in this section is provided
with the \textsc{knockoff\_matlab} package
for \textsc{Matlab},
available at \url
{http://web.stanford.edu/\textasciitilde candes/Knockoffs/package_matlab.html}.}
The data set, described and analyzed in
\cite{rhee2006genotypic},
consists of drug resistance measurements and genotype information from
samples of HIV-1, with separate data sets for resistance to protease
inhibitors (PIs), to nucleoside reverse transcriptase (RT) inhibitors
(NRTIs) and to nonnucleoside RT inhibitors (NNRTIs). The data set
sizes are as follows:\vspace*{12pt}

{\fontsize{9}{11}\selectfont
\begin{center}
\tabcolsep=0pt
\begin{tabular*}{\textwidth}{@{\extracolsep{\fill}}lcccc@{}}
\hline
 &  & & \textbf{\# protease or RT}
&\textbf{\# mutations
appearing} \\
\textbf{Drug type}&\textbf{\# drugs}&\textbf{Sample size}&
\textbf{positions genotyped}&$\bolds{\geq\!3}$ \textbf{times in sample}\\
\hline
PI&7&846&\phantom{0}99&209\\
NRTI&6&634&240&287\\
NNRTI&3&745&240&319\\
\hline
\end{tabular*}\vspace*{12pt}
\end{center}}
\noindent In each drug class, some samples are missing
resistance measurements for some of the drugs, so for each drug's
analysis, the sample size and the number of mutations present are
slightly smaller than given in the table; we report the final $n$ and $p$
for each drug in Figures \ref{fig:HIV_PI}, \ref{fig:HIV_NRTI} and
\ref{fig:HIV_NNRTI}
on a case-by-case basis.

We analyze each drug separately. The response $y_i$ is given
by the log-fold-increase of lab-tested drug resistance in the $i$th
sample, while the design matrix $\X$ has entries $X_{ij}\in\{0,1\}$,
indicating presence or absence of mutation \#$j$ in the $i$th sample.
(For each drug, we keep only those mutations appearing $\geq\!3$ times
in the sample for that drug, and we remove duplicated columns from $\X
$ to allow
for identifiability.)
Different mutations at the same position
are treated as distinct features, and we assume an additive linear
model with no interactions. We then apply knockoff and BHq, each with
$q=20\%$, to
the resulting data set. One of the drugs has a sample size
$n$ with $p<n<2p$, in which case
we use the method described
in Section~\ref{sec:p_or_2p} which extends the knockoff method
beyond the original construction for the $n\geq2p$
regime; see Figures \ref{fig:HIV_PI}, \ref{fig:HIV_NRTI} and \ref
{fig:HIV_NNRTI}
for the values of $n$ and~$p$.

\begin{figure}

\includegraphics{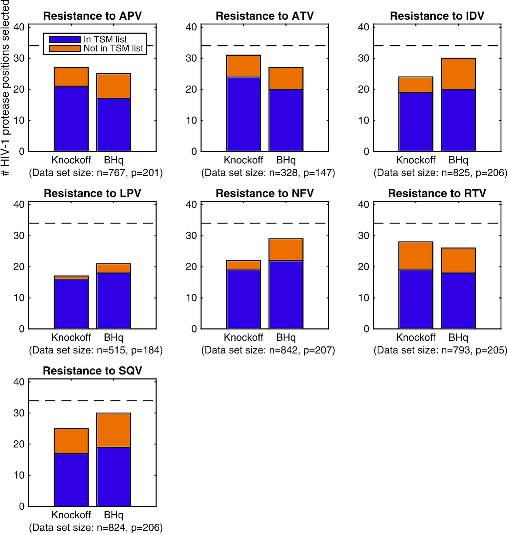}

\caption{Results of applying the knockoff filter and BHq with $q=20\%$ to
model PI-type drug resistance of HIV-1 based on genetic mutations
using data from \cite{rhee2006genotypic}. For each PI-type treatment
and for each of the three methods,
the bar plots show
the number of positions on the HIV-1 protease where mutations were
selected. To
validate the selections of the methods, dark blue
indicates protease positions that appear in the treatment-selected
mutation (TSM) panel
for the PI class of treatments,
given in Table~1 of \cite{rhee2005hiv},
while orange indicates positions selected by the method that do not
appear in the TSM list.
The horizontal line indicates the total number of HIV-1 protease
positions appearing in the TSM list.
Note that the TSM list consists of mutations that are associated with
the PI class of drugs in general,
and is not specialized to the individual drugs in the class.}
\label{fig:HIV_PI}
\end{figure}

\begin{figure}

\includegraphics{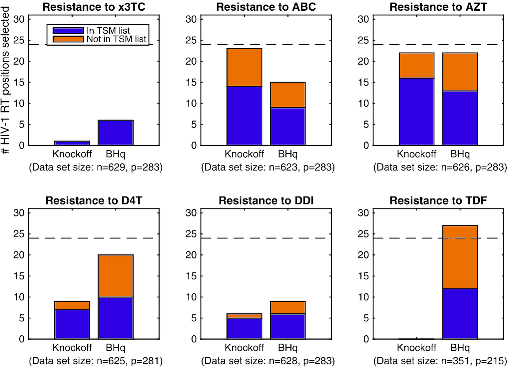}

\caption{Same as Figure \protect\ref{fig:HIV_PI}, but for the NRTI-type
drugs, validating
results against the treatment-selected mutation (TSM) panel for NRTIs given
in Table~2 of \cite{rhee2005hiv}.}
\label{fig:HIV_NRTI}
\end{figure}

\begin{figure}[b]

\includegraphics{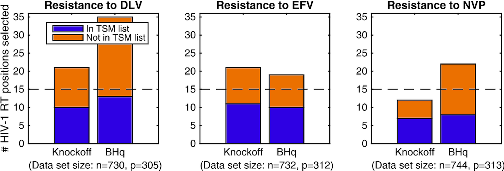}

\caption{Same as Figure \protect\ref{fig:HIV_PI}, but for the NNRTI-type
drugs, validating
results against the treatment-selected mutation (TSM) panel for NNRTIs given
in Table~3 of \cite{rhee2005hiv}.}\label{fig:HIV_NNRTI}
\end{figure}

To evaluate the results, we compare the selected mutations with
existing treatment-selected mutation (TSM) panels \cite{rhee2005hiv};
since this is a real data experiment, the ground truth is unknown, but
these panels provide a good approximation that we can use to assess
the methods. For each drug class (PIs, NRTIs, NNRTIs), Rhee et al.
\cite{rhee2005hiv} create panels of mutations that are present at
significantly higher frequency (after correcting for multiple
comparisons) in virus samples from individuals who have been treated
with that class of drug, as compared to individuals that have never
received that class of drug. Therefore the data we use for model
selection (based on lab-tested drug resistance) and the mutation
panels used to validate our results (based on association with patient
treatment history) come
from different types of studies, and we aim to see replicability; that
is, we will evaluate our model selection results based on how many of
the mutations identified by our analysis appear also in the TSM lists.
It is known that multiple mutations at the same protease or RT
position can often be associated with related drug-resistance
outcomes. Since the TSM lists
are an approximation of the ground truth, we will compare only
the positions of the mutations selected, with the positions of
mutations on the TSM lists.

Results for the PI, NRTI and NNRTI type drugs are displayed in
Figures \ref{fig:HIV_PI}, \ref{fig:HIV_NRTI} and~\ref{fig:HIV_NNRTI},
respectively. We see that both methods perform similarly for most of
the drugs in the three classes, with good agreement in most cases
between the positions of the selected mutations based on the
lab-tested drug resistance data, and the TSM lists which are based on
patient history data. Overall, the knockoff filter shows slightly
better agreement with the TSM lists as compared to BHq, but there is
variability in the outcomes across the different drugs. In summary we
see that the FDR-controlling knockoff methods indeed select variables
that mostly correspond to real (replicable) effects, as verified by
the independently created TSM lists.

\subsection{Simulation under non-Gaussian noise with a sparse real
design matrix}\label{sec:HIV_simulation}

To verify that the knockoff method is robust to non-Gaussian noise, we
test the method (and compare to BHq) using partially-simulated data.
First, we take the design matrix
$\X\in\R^{747\times319}$ from the NNRTI-drug data set discussed
above. We generate a coefficient
vector $\beta\in\R^{319}$ by randomly choosing a support $S$ of size $k=20$
and drawing $\beta_j\sim3.5 \cdot N(0,1)$ for each $j\in S$ (and
$\beta_j=0$ for $j\notin S$).

To obtain a realistic non-Gaussian distribution on the noise, we
consider the empirical distribution
given by $\mathcal{P}_{{\X}}^{\perp}({\y^{(l)}})$ where $\y^{(l)}$
is the response for
the $l$th NNRTI-type drug for $l=1,2,3$
(discarding rows of $\X$ if entries of $\y^{(l)}$ are missing). This
empirical noise distribution
is heavy tailed. (Its sample excess kurtosis is equal to $9.24$, which
is the population
excess kurtosis for a $t$ distribution with $4.65$ degrees of freedom.)
To generate the data, we set
$\y= \X\beta+ {\mathbf z}$
where the entries $z_i$ of $\mathbf z$ are sampled with replacement
from the empirical noise distribution after rescaling to ensure that
$\mathbb{E}[{z_i^2}]=1$.
The following results show that the knockoff,
knockoff$+$ and BHq procedures each exhibit good FDR control and power:\vspace*{12pt}

{\fontsize{9}{11}\selectfont
\begin{center}
\begin{tabular*}{\textwidth}{@{\extracolsep{\fill}}lcc@{}}
\hline
&\textbf{FDR over 1000 trials} & \textbf{Power over 1000 trials}\\
& \textbf{(nominal level} $\bolds{q=20\%}$\textbf{)}\\
\hline
Knockoff&25.72\%&63.50\%\\
Knockoff$+$&20.31\%&60.67\%\\
BHq&25.47\%&69.42\%\\
\hline
\end{tabular*}\vspace*{12pt}
\end{center}}
\noindent We note that this setting, where the columns of $\X$ are sparse and
the noise is heavy tailed,
is an extremely challenging scenario for the knockoff procedure to
maintain FDR control. For instance, even with no signals present ($\b
=0$), the marginal statistic
$\X_j^\top\y=\X_j^\top\mathbf z=\sum_i X_{ij} z_i$ follows a highly
non-Gaussian\vspace*{1pt} distribution; the
central limit theorem cannot be applied
to this sum because $\X_j$ is highly sparse, and the $z_i$'s come from
a heavy-tailed
distribution.
Nonetheless, we see approximate FDR control in the results of this simulation.

\section{Sequential hypothesis testing}\label{sec:sequential}

\subsection{Two sequential testing procedures}

In this section, we describe several related sequential hypothesis testing
procedures, along with theoretical results for FDR control. We then
relate these procedures to the knockoff and knockoff$+$ methods, in order
to prove our main results, Theorems \ref{thm:knockoff} and \ref
{thm:knockoffplus}.

Imagine that $p_1,\ldots,p_m$ are $p$-values giving information about
hypotheses $H_1,
\ldots,
H_m$. These $p$-values
obey $p_j\geqd\mathsf{Unif}[0,1]$ for all null $j$;
that is, for all null $j$ and all $u \in[0,1]$, $\mathbb{P}\{{p_j \le
u}\} \le
u$.
We introduce two sequential strategies, which control the FDR at any
fixed level $q$ under a usual independence property.

\textit{Sequential step-up procedure (SeqStep)}. Fix any
threshold $c\in(0,1)$ and any subset\footnote{In many applications
we would typically choose $K=[m]$, but we allow for
$K\subsetneq[m]$ to help with the proof of the regression method.}
$K$, and define
\[
\hat{k}_{0} = \max \biggl\{k\in K: \frac{\#\{j\leq k:p_j>c\}}{
k\vee1}\leq(1-c)\cdot q
\biggr\}
\]
and
\[
\hat{k}_{1} = \max \biggl\{k\in K: \frac{1 + \#\{j\leq k:p_j>c\}}{1 +
k}\leq(1-c)\cdot q
\biggr\},
\]
with the convention that we set $\hat{k}_{0/1} =0$ if the set is
empty: here $\hat{k}_{0/1}$ should be read as ``$\hat{k}_0$ or
$\hat{k}_1$'' since we can choose which of the two definitions above
to use. We then reject $H_j$ for all $j\leq\hat k_{0/1}$, and thus
get two distinct procedures named SeqStep (using $\hat k_0$)
and SeqStep$+$ (using $\hat k_1$) hereafter.

To understand why such a sequential procedure makes sense, consider
SeqStep, and assume that the null $p$-values are
i.i.d. $\mathsf{Unif}[0,1]$. Then
\[
\frac{\#\{\mbox{null } j\leq k\}}{
k\vee1} \approx\frac{1}{1-c} \cdot\frac{\#\{\mbox{null } j\leq
k: p_j>c\}} {
k\vee1} \le
\frac{1}{1-c} \cdot\frac{\#\{j\leq k: p_j>c\}} {
k\vee1}
\]
so that again, the procedure maximizes the number of rejections under
the constraint that an estimate of FDR is controlled at level
$q$. SeqStep$+$ corrects SeqStep to guarantee FDR control.

\textit{Selective sequential step-up procedure (Selective SeqStep)}.
Alternatively,
define
\[
\hat{k}_{0/1}=\max \biggl\{k\in K: \frac{0/1+\#\{j\leq k:p_j>c\}}{
\#\{j\leq k: p_j\leq c\}\vee1}\leq
\frac{1-c}{c}\cdot q \biggr\},
\]
with
the convention that we set $\hat{k}_{0/1} = 0$ if this set is
empty. (We get two distinct procedures named Selective SeqStep and
Selective SeqStep$+$ by
letting the term in the numerator be 0 or 1.) Then reject $H_j$ for all
$j\leq\hat{k}_{0/1}$
such that $p_j\leq c$. Strictly speaking, these are not sequential testing
procedures (because among the first $\hat{k}_{0,1}$
hypotheses in the list, we reject only those satisfying the selective
threshold $p_j\leq c$),
and we are thus abusing terminology.

Again, to understand this procedure intuitively when the null
$p$-values are
i.i.d. $\mathsf{Unif}[0,1]$, we see that
\begin{eqnarray*}
\frac{\#\{\mbox{null } j\leq k : p_j \le c\}}{
\#\{j\leq k: p_j\leq c\}\vee1}& \approx&\frac{c}{1-c} \cdot\frac{\#\{
\mbox{null } j\leq k : p_j>c\}} {
\#\{j\leq k: p_j\leq c\}\vee1}\\
& \le&
\frac{c}{1-c} \cdot\frac{\#\{j\leq k: p_j>c\}} {
\#\{j\leq k: p_j\leq c\}\vee1}
\end{eqnarray*}
so that again, the procedure maximizes the number of rejections under
the constraint that an estimate of FDR is controlled at level
$q$.

\begin{theorem}\label{thm:seq}
Suppose that the null $p$-values are i.i.d. with $p_j \geq
\mathsf{Unif}[0,1]$, and are independent from the nonnulls.
For each procedure considered, let
$V$ be the number of false discoveries and $R$
the total number of discoveries:
\begin{itemize}
\item Both SeqStep$+$ and Selective SeqStep$+$ control the FDR, that is,
$
\mathbb{E}[{\frac{V}{ R\vee1}}] \le q$.
\item Selective SeqStep controls a modified FDR,
$
\mathbb{E}[{\frac{V}{R+({c}/{(1-c)}) q^{-1} }}] \le q$.

\item SeqStep also controls a modified FDR,
$
\mathbb{E}[{\frac{V}{R+({1}/{(1-c)}) q^{-1}}}] \le q$.
\end{itemize}
\end{theorem}

As is clear from the assumption, the order of the $p$-values cannot be
dependent on the $p$-values themselves---for instance, we cannot
reorder the $p$-values from smallest to largest, apply this procedure
and expect FDR control.

\subsection{Connection with knockoffs}
\label{sec:theorem1}

Interestingly, the knockoff method can be cast as a special case of the second
sequential hypothesis testing procedure, and the FDR controlling
properties are then just a consequence of Theorem~\ref{thm:seq}. We explain
this connection, thereby proving Theorems \ref{thm:knockoff} and
\ref{thm:knockoffplus}.

Let $m=\#\{j:W_j\neq0\}$; since our method never selects variable $j$
when $W_j=0$,
we can ignore such variables.
Assume without loss of generality that $|W_1|\geq|W_2|\geq\cdots\geq
|W_m|>0$, and set
\[
p_j = \cases{ 1/2, &\quad $W_j > 0$,\vspace*{2pt}
\cr
1, &\quad
$W_j < 0$,}
\]
which can be thought of as 1-bit $p$-values. It then follows from
Lemma~\ref{lem:key} that the null $p$-values are i.i.d. with $\mathbb
{P}\{{p_j = 1/2}\} = 1/2 = \mathbb{P}\{{p_j = 1}\}$ and are
independent from the others,
thereby obeying the assumptions of Theorem~\ref{thm:seq}. Setting $K$
to be the
indices of the strict inequalities,
\[
K = \bigl\{k\in[m]: |W_k|>|W_{k+1}| \bigr\} \cup\{m\},
\]
one sees that the knockoff method is now
equivalent to the second sequential testing procedure on these
$p$-values. To see why this true, set $c = 1/2$, and observe that for
any $k\in K$,
\[
\frac{0/1+\#\{j\leq k:p_j>1/2\}}{ \#\{j\leq k: p_j\leq1/2\}\vee1} = %
\frac{0/1+\#\{j: W_j \leq-|W_k|\}}{\#\{j: W_j \geq
|W_k|\}\vee1}.
\]
%
Hence, finding the largest $k$ such that the ratio in the left-hand
side is below $q$ is the same as finding the smallest $|W_k|$ such
that the right-hand side is below $q$,
which is equivalent to calculating the knockoff or knockoff$+$ threshold $T$
given in \eqref{eqn:thresh} or \eqref{eqn:threshplus}, respectively.
Finally, rejecting
the $p$-values obeying $p_j \le1/2$ is the same as rejecting the
positive $W_j$'s. FDR control follows by applying Theorem~\ref{thm:seq}.

\section{Discussion}
\label{sec:discussion}
In this paper, we have proposed two variable selection procedures,
knockoff and knockoff$+$, that control FDR in the linear regression
setting and offer high power to discover true signals. We give
theoretical results showing that these methods maintain FDR control
under arbitrary feature correlations, even when variable selection methods
such as the Lasso may select null variables far earlier than some of
the weaker signals. The empirical performance of knockoff and knockoff$+$
demonstrates effective FDR control and excellent power in comparison
to other methods such as the Benjamini--Hochberg procedure (BHq) or
permutation-based methods.

A key ingredient in the knockoff and knockoff$+$ methods is the
``one-bit $p$-values'' obtained by comparing feature $\X_j$ with its
knockoff feature $\Xproxy_j$, and recording which of the two was first
to enter the Lasso path. This extreme discretization may be part of
the reason that the methods are conservative under low signal
amplitude, and could potentially be addressed by creating multiple
knockoffs $\Xproxy_j^{(1)}$, \ldots, $\Xproxy_j^{(m)}$ for each feature
$\X_j$. We will investigate the potential
benefits of a multiple-knockoff approach in future work.

In the theoretical portions of this paper,
we have entirely focused on controlling a Type-I error, namely, the
FDR; this paper did not study the false negative rate (FNR), defined as
the expected fraction of nondetected hypotheses among the true
nonnulls. This is a question of statistical power, which we have
demonstrated empirically, and we leave a theoretical analysis to future work.

Finally, the analysis and methods presented here rely on
the assumption that $\bolds\Sigma=\X^\top\X$ is invertible, which is
necessary so that $\X_j$ does not lie in the span of the remaining
$(p-1)$ features, and its effect on the response is, therefore,
identifiable.\footnote{When this is not the case, we cannot
distinguish between a mean function that depends on $\X_j$ versus
a mean function that depends on the linear combination of the
other variables.} In many modern applications, however, we are
interested in a regime where $p>n$ and $\bolds\Sigma$ is defacto
noninvertible.
Here, there are many types of common additional assumptions that may
allow us to overcome the identifiability problem---for instance,
sparse dependence structure among the features themselves and
between the mean response and the features. One possible approach
is to split the observations in two disjoint sets, using the first
to screen for a smaller set of potential features, and the second to
run the (low-dimensional) knockoff filter over this smaller set of
features only. Our ongoing research develops knockoff methodology
for high dimensions with higher power, and we certainly hope to
report on our findings in a future publication.

\section*{Acknowledgments}
E. Cand\`es would like to thank Ma{\l}gorzata Bogdan, Lucas
Janson, Chiara Sabatti and David Siegmund for helpful discussions and
useful comments about an early version of the manuscript. E. Cand\`es also
thanks Fr\'ed\'eric Benqu\'e for a productive brainstorming session
about names, Carlos Fernandez-Granda for his help in running
simulations on the Solomon cluster and Jonathan Taylor for suggesting
the analysis of the HIV dataset. We are grateful to Soo-Yon Rhee and
Rob Shafer for discussing the HIV dataset with us.

\begin{supplement}[id=suppA]
\stitle{Supplement to ``Controlling the false discovery rate via knockoffs''}
\slink[doi]{10.1214/15-AOS1337SUPP} 
\sdatatype{.pdf}
\sfilename{aos1337\_supp.pdf}
\sdescription{We provide details for the proofs of several theoretical
results in the paper.}
\end{supplement}

%

%




\printaddresses
\end{document}